\documentclass[journal]{IEEEtran}
\usepackage{amsmath,amssymb,amsfonts}
\usepackage{algorithmic}
\usepackage{graphicx}
\usepackage{textcomp}
\usepackage{xcolor}
\usepackage{bm}
\usepackage{cite}
\usepackage{tikz}
\usetikzlibrary{fit,tikzmark}
\usepackage{graphicx}
\usepackage{subcaption}

\newtheorem{theorem}{Theorem}

\newtheorem{corollary}{Corollary}

\begin{document}

\title{Two-dimensional Fourier compressed sensing\\ under a fixed readout budget per channel}

\author{Nitin Jonathan Myers,~\IEEEmembership{ Senior Member,~IEEE}\\
Delft Center for Systems and Control, Delft University of Technology
\vspace{-4mm}
}

\maketitle

\begin{abstract}
Recovering sparse signals from their subsampled Fourier representation is an important problem in communications, radar, and imaging. In this letter, we focus on  reconstructing sparse 2D signals (matrices) under the constraint that only a fixed number of entries can be sampled from each channel, e.g., a row or a column in the Fourier domain. For a specified per-channel readout budget, we derive a lower bound on the mutual coherence of the corresponding compressed sensing  matrix. We show that our bound is larger than the classical Welch bound, due to a limited readout budget. We also construct deterministic subsampling patterns that attain this bound for a class of matrix dimensions and readout budgets, and benchmark them against random subsampling through simulations.
\end{abstract}


\section{Introduction}
Sparse signals can be recovered from fewer samples in their Fourier representation using compressed sensing (CS) \cite{donoho2006compressed}. In some 2D-CS settings, hardware constraints impose a readout limit per channel. For example, analog correlator-based radars acquire code-domain measurements using correlator banks \cite{undavalli2025fully, zhou2022band}. In a MIMO analog correlator-based radar, chip-area constraints can limit the number of physical analog correlators per receiver and hence the number of measurements acquired per receive channel \cite{vardakis2026correlator}. Similar per-channel readout limits arise in stepped-frequency radar and imaging systems \cite{mishra2020range,oike2012cmos}.

\par The measurement pattern, called subsampling pattern in CS, directly influences the reconstruction quality. For instance, the subsampling pattern determines the mutual coherence of the CS matrix; lower coherence reduces the upper bound on the reconstruction error \cite{ben2010coherence} with the orthogonal matching pursuit (OMP) \cite{tropp2007signal}. Therefore, prior work on Fourier CS in \cite{applebaum2009chirp,xu2014compressed} has designed subsampling patterns to minimize mutual coherence, whose fundamental limit is given by the Welch bound \cite{welch1967lower}. For CS of 1D signals (vectors), special partial Fourier sensing matrices based on cyclic difference sets \cite{Welch_diff_sets} and almost difference sets \cite{yu2014new, yu2013deterministic} are known to achieve or approach the Welch bound; these matrices exist only for certain combinations of the vector's dimension and the total measurement budget. When such constructions do not exist, numerical methods to minimize coherence-related objectives were developed in \cite{rusu2018algorithms,ariananda2016deterministic}. Recent work has also considered 2D subsampling pattern design using difference sets \cite{campman2017sparse} and data-driven methods \cite{hernandez2023design,bahadir2019learning,ravula2023optimizing, sherry2020learning}. These works, however, do not address the per-channel readout constraint considered in this paper. The closest related work is \cite{kumari2021adaptive}, which designs 2D subsampling patterns under a fixed readout constraint, but only for the special case of one readout per channel.

\par This paper generalizes the construction in \cite{kumari2021adaptive} to the case where multiple measurements can be acquired from each channel. Specifically, we consider the recovery of a $P\times Q$ matrix, whose 2D-discrete Fourier transform (2D-DFT) is sparse, when only $K$ distinct entries can be sampled from each row. Our contributions are listed below.
\begin{itemize}
\item We derive a tight lower bound on the mutual coherence of the CS matrix under a per-channel readout budget and show that this bound is strictly larger than the classical Welch bound for a comparable total measurement budget.
\item We construct deterministic subsampling patterns that attain the proposed mutual-coherence bound for a specified per-channel readout budget. Our construction applies when $P=Q$, $P$ is prime, and $Q-1$ divides $K(K-1)$, as it relies on cyclic difference sets \cite{baumert2006cyclic} and prime-length Zadoff-Chu sequences \cite{chu1972polyphase}.
\item We show, using simulations, that the subsampling patterns constructed in this paper achieve a lower signal reconstruction error than comparable random patterns that satisfy the same per-channel readout budget.
\end{itemize}
\begin{figure}[h]
    \centering
    \vspace{-3mm}

    \begin{minipage}{0.48\linewidth}
        \centering
        \includegraphics[width=0.65\linewidth]{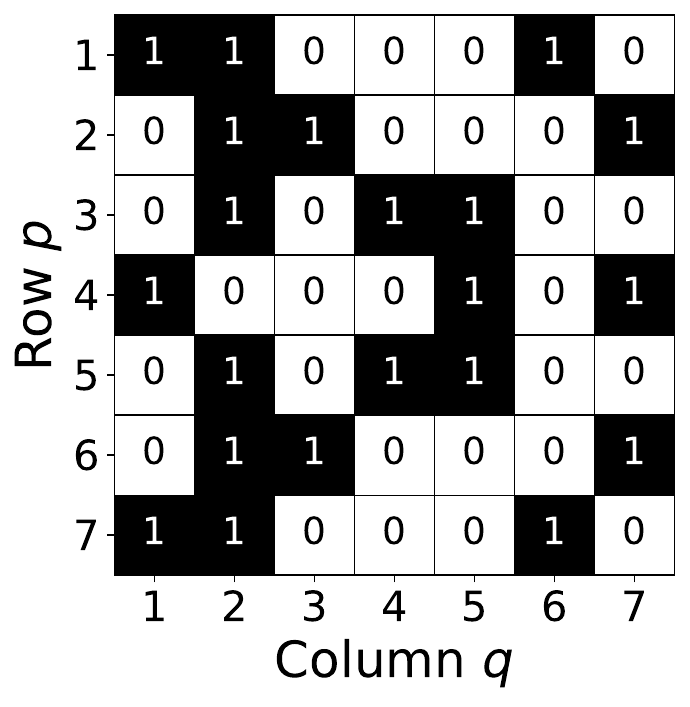}
        \caption*{(a) Readout for a per-row sampling limit of $K=3$.}
    \end{minipage}
    \hfill
    \begin{minipage}{0.48\linewidth}
        \centering
        \includegraphics[width=0.65\linewidth]{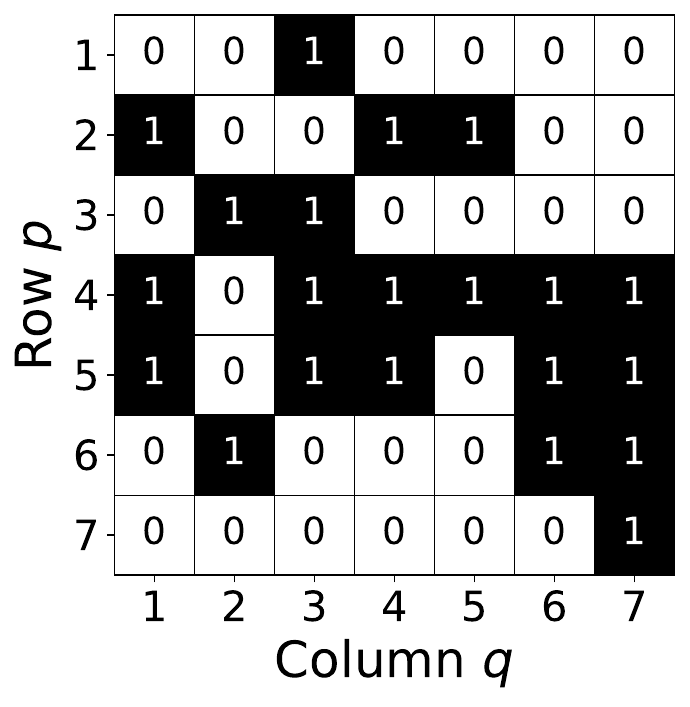}
        \caption*{(b) Unconstrained readout of $M=21$ samples in total.}
    \end{minipage}

    \caption{Examples showing the acquisition of $M=21$ measurements from a $7 \times 7$ matrix. We focus on readout schemes of the type shown in (a), which acquire only $K$ entries from each row, and design subsampling (readout) patterns to achieve the smallest possible mutual coherence for 2D-Fourier CS.}
    \label{fig:sampling_structures}
    \vspace{-5mm}
\end{figure}
\section{2D-Fourier CS under a readout budget}
We consider a sparse signal $\mathbf{X}\in\mathbb{C}^{P \times Q}$, a $P\times Q$ matrix, with $s$ non-zero entries. The 2D-DFT of $\mathbf{X}$ is defined as 
\begin{equation}
    \mathbf{H}= \mathbf{U}_{P}\mathbf{X}\mathbf{U}_Q,
\end{equation}
where $\mathbf{U}_P$ and $\mathbf{U}_Q$ denote unitary DFT matrices of size $P\times P$ and $Q\times Q$. In 2D-Fourier CS, the Fourier representation of the sparse signal is subsampled. The $m^{\mathrm{th}}$ measurement in a sequence of $M$ measurements is defined as  
\begin{align}
\label{eq:gn_Fourier_meas}
    y[m] &= H(r[m], c[m]) + v[m]\\
    &= \mathbf{e}^T_{r[m], P} \mathbf{U}_P \mathbf{X}\mathbf{U}_Q \mathbf{e}_{c[m], Q} + v[m],
\end{align}
where $(r[m], c[m])$ denotes the sampled coordinate on the $P\times Q$ 2D-discrete frequency grid, $v[m]$ is additive noise, and $\mathbf{e}_{k,J}$ is the $k^{\mathrm{th}}$ column of a $J\times J$ identity matrix. The set of sampled coordinates is defined as $\Omega=\{(r[m], c[m])\}_{m=1}^{M}$. The vector $\mathbf{y}$ of $M$ measurements is then  
\begin{align}
\label{eq:Fourier_meas_vec}
    \mathbf{y}&=\mathbf{Ax}+\mathbf{v}, \\
    &=\mathcal{F}_{\mathrm{2D,}\Omega}(\mathbf{X})+\mathbf{v}.
\end{align}
where the $m^{\mathrm{th}}$ row of the CS matrix $\mathbf{A}$ is $\mathbf{A}(m,:)=(\mathbf{U}_Q \mathbf{e}_{c[m],Q})^T \otimes (\mathbf{e}^T_{r[m],P} \mathbf{U}_P)$. Furthermore,  $\mathcal{F}_{\mathrm{2D,}\Omega}(\mathbf{X})$ is the subsampled 2D-DFT operator and $\mathbf{x}=\mathrm{vec}(\mathbf{X})$. CS algorithms estimate $\mathbf{X}$ from $\mathbf{y}$, which contains $M<PQ$ measurements. 
\par We consider a readout constraint that permits the acquisition of only $K$ distinct samples from each row of $\mathbf{H}$. Since $\mathbf{H}$ has $P$ rows, the total number of measurements in our setup is
\begin{equation}
M=PK.
\end{equation}
Fig. \ref{fig:sampling_structures}(a) illustrates a subsampling pattern under the per-channel budget constraint, while Fig. \ref{fig:sampling_structures}(b) shows an unconstrained random pattern for the same total number of measurements. To model the subsampling configuration $\Omega$, we define a binary matrix $\mathbf{B} \in \{0,1\}^{P \times Q}$, where 
\begin{equation}
B(p,q)=
\begin{cases}
1, & \text{if } (p,q)\in\Omega,\\
0, & \text{otherwise}.
\end{cases}
\end{equation}
The 2D-DFT of $\mathbf{B}$, denoted by $\mathcal{F}_{\mathrm{2D}}(\mathbf{B})$, is related to the point spread function ($\mathrm{PSF}$) as \cite{lustig2007sparse}
\begin{align}
\label{eq:PSF_defn}
\mathrm{PSF}&= \sqrt{\frac{Q}{K^2P}}\mathbf{U}_P \mathbf{B}\mathbf{U}_Q.
\end{align}
The $\mathrm{PSF}$ in 2D-DFT CS determines the success of sparse recovery. For example, the first iteration of the OMP identifies the strongest entry in $\mathbf{A}^{\ast}\mathbf{Ax}+\mathbf{A}^{\ast}\mathbf{v}$, which is also equal to $\mathrm{vec}(\mathrm{PSF} \circledast \mathbf{X}) + \mathbf{A}^{\ast}\mathbf{v}$ \cite{myers2019falp}, where $(\cdot)^{\ast}$ denotes the conjugate-transpose and $\circledast$ the 2D-circular convolution. The ideal $\mathrm{PSF}=\mathbf{e}_{1,P}\mathbf{e}^T_{1,Q}$ preserves the entries in $\mathbf{x}$ even after circular convolution, enabling perfect support detection. Such a $\mathrm{PSF}$, however, is  achieved only when $\mathbf{B}$ is a matrix of all ones, equivalently when the readout budget $K=Q$. In practice, $K$ can be much smaller than $Q$ due to hardware limits. For example, a 2D MIMO radar with only $K$ correlators per receiver can acquire only $K$ entries from each row of the 2D virtual-array matrix \cite{sun2015mimo}. When $K\ll Q$, the subsampling matrix $\mathbf{B}$ has only $PK$ number of ones and $\mathrm{PSF}\neq \mathbf{e}_1\mathbf{e}^T_1$. 
\par To minimize the upper bound on the reconstruction error with the OMP \cite{ben2010coherence}, $|\mathrm{PSF}(i,j)|$ must be made as small as possible for all $(i,j) \neq (1,1)$. The mutual coherence of the CS matrix $\mathbf{A}$ is defined as $\mu=\mathrm{max}_{(i,j)\neq (1,1)} |\mathrm{PSF}(i,j)|$ \cite{kumari2021adaptive}. Given a per-channel readout budget, our goal is to determine the subsampling pattern that achieves the smallest $\mu$, i.e.,
\begin{equation}
\label{eq}
\mathcal{O}:\left\{
\begin{aligned}
  \mathbf{B}_{\mathrm{opt}}
  =&\operatorname*{arg\,min}_{\mathbf{B}\in\{0,1\}^{P\times Q}}\;\;\;
    \max_{(i,j)\neq(1,1)}
    \left|\operatorname{PSF}(i,j)\right|  \\
 &\mathrm{s.t.}
  \sum_{q=1}^{Q} B(p,q)=K,\quad p=1,\ldots,P\,.
\end{aligned}
\right.
\end{equation}
In this paper, we derive a lower bound on $\mu$ under the per-channel readout budget and  provide a closed-form construction of $\mathbf{B}_{\mathrm{opt}}$ when $P=Q$, $P$ is prime, and $Q-1$ divides $K(K-1)$. 

\section{Mutual coherence lower bound\\ under a per-channel readout budget}
For an unconstrained CS matrix of size $M\times N$, the lower bound on the mutual coherence is given by the Welch-bound 
$\mu_{\mathrm{wb}}= \sqrt{(N-M)/(MN-M)}$ \cite{xu2014compressed}. This bound depends only on the total number of measurements $M$ and the ambient signal dimension $N$. As $M=PK$ and $N=PQ$ in our setup, the  Welch-bound in our CS problem evaluates to  
\begin{align}
\label{eq:welch_bound_simp}
\mu_{\mathrm{wb}}&= \sqrt{\frac{Q-K}{KPQ-K}}.
\end{align}
A subsampling pattern that achieves this bound corresponds to a $\mathrm{PSF}$ whose largest off-peak magnitude, i.e., the largest value over all indices other than $(1,1)$, is as small as possible.

\par The per-channel readout budget enforces a new constraint on the PSF, due to which the classical Welch bound cannot be attained. To explain this, we construct the $\mathrm{PSF}$ in \eqref{eq:PSF_defn} in two stages: (i) Compute $\tilde{\mathbf{B}}=\mathbf{B}\mathbf{U}_Q$, i.e., the 1D-DFT of $\mathbf{B}$ along its rows and (ii) Compute $\mathrm{PSF}=\sqrt{\frac{Q}{K^2P}} \mathbf{U}_P\tilde{\mathbf{B}}$, which contains the scaled 1D-DFT of the columns in $\tilde{\mathbf{B}}$. As every row of a feasible $\mathbf{B}$ has exactly $K$ ones and $Q-K$ zeros, the first entry of its 1D-DFT is $K/\sqrt{Q}$, i.e., $\tilde{{B}}(p,1)=K/\sqrt{Q}\,\forall p$. Next, as $\tilde{{B}}(p,1)=K/\sqrt{Q}\,\forall p$, the first column in $\mathbf{U}_P \tilde{\mathbf{B}}$ evaluates to  $KP\mathbf{e}_1/\sqrt{QP}$. Therefore, the readout budget forces the first column of $\mathrm{PSF}$ to $(1,0,0,\cdots,0)^T$ for any feasible subsampling configuration $\mathbf{B}$.

\par The configuration within $\mathbf{B}$ determines the entries of the $\mathrm{PSF}$ outside its first column. The mutual coherence $\mu$ is thus the largest magnitude among these entries. We derive a lower bound on $\mu$ by leveraging that the first column of $\mathrm{PSF}$ is $(1,0,0,\cdots,0)^T$ for any feasible readout scheme.
\begin{theorem}[Proposed coherence bound]\label{thm:bound} Under a per-channel readout budget of $K$ entries per row, the mutual coherence of the 2D-DFT CS matrix $\mu$ is lower bounded as 
\begin{equation}\label{eq:bound_prop}
\mu \geq \mu_{\mathrm{const}, K}\ =\ \sqrt{\dfrac{Q-K}{KPQ-KP}}\;.
\end{equation}
\end{theorem}
\begin{IEEEproof}
Our proof uses Parseval's theorem, which equates the total energy in $\mathbf{B}$, i.e., $\Vert \mathbf{B}\Vert^2_{\mathrm{F}}=KP$, to that of its unitary 2D-DFT $\mathcal{F}_{\mathrm{2D}}(\mathbf{B})$. Note that $\mathrm{PSF}=\sqrt{Q/(K^2P)}\mathcal{F}_{\mathrm{2D}}(\mathbf{B})$. We know that the $(1,1)^{\mathrm{th}}$ entry in $\mathcal{F}_{\mathrm{2D}}(\mathbf{B})$, i.e., the DC-component is $KP/ \sqrt{PQ}$. Furthermore, the remaining $P-1$ entries in the first column of $\mathcal{F}_{\mathrm{2D}}(\mathbf{B})$ are zero. We apply Parseval's theorem to write 
\begin{align}
    \Vert\mathbf{B}\Vert^2_{\mathrm{F}}&=\left(\frac{KP}{\sqrt{PQ}}\right)^2\!\!\!+\sum^{P}_{p=2}0^2\! +\! \sum_{p=1}^P\sum_{q=2}^Q \left|[\mathcal{F}_{\mathrm{2D}}(\mathbf{B})]_{p,q}\right|^2\\
    \Rightarrow &KP=\frac{K^2P}{Q}+ \frac{K^2P}{Q} \sum_{p=1}^P\sum_{q=2}^Q |\mathrm{PSF}(p,q)|^2\\
    \Rightarrow &KP \overset{(a)}{\leq} \frac{K^2P}{Q}+ \frac{K^2P}{Q} P(Q-1)\mu^2\\
    \Rightarrow &\mu \geq \sqrt{\dfrac{Q-K}{KPQ-KP}}.
\end{align}
Here, $(a)$ follows from the fact that $\mu=\mathrm{max}_{(i,j)\neq (1,1)} |\mathrm{PSF}(i,j)|$. 
\end{IEEEproof}
 \begin{corollary}\label{cor:strictly_higher} For $K<Q$, our coherence bound under the per-channel readout is strictly larger than the Welch bound for the same total number of samples. To see this, we take the ratio between the bound in \eqref{eq:bound_prop} and the Welch bound in \eqref{eq:welch_bound_simp}: \begin{equation} \label{eq:mu_bound_ratio} \frac{\mu_{\mathrm{const},K}}{\mu_{\mathrm{wb}}} = \sqrt{\frac{PQ-1}{PQ-P}} . \end{equation} Since $P>1$ in a 2D-CS setting, we have $PQ-1>PQ-P$, and hence $\mu_{\mathrm{const},K}/\mu_{\mathrm{wb}}>1$. Thus, the per-channel readout constraint raises the lower bound on the smallest achievable coherence relative to unconstrained readout. 
 \end{corollary}
 \begin{corollary}\label{cor:square_geom_asymptotics} For CS of square matrices, i.e., $P=Q$, our bound approaches the unconstrained Welch bound as $P\rightarrow\infty$. To prove this, we substitute $P=Q$ in \eqref{eq:mu_bound_ratio} to write \begin{equation} \frac{\mu_{\mathrm{const},K}}{\mu_{\mathrm{wb}}} = \sqrt{\frac{P^2-1}{P^2-P}} = \sqrt{\frac{P+1}{P}}, \end{equation} which converges to $1$ as $P\rightarrow\infty$. Therefore, for recovery of square matrices, the coherence bound penalty introduced by the per-channel readout budget vanishes asymptotically.\end{corollary}
\section{Constructions achieving the smallest mutual \\ coherence under a per-channel readout budget}\label{sec:constr}
\par To explain our construction, we consider $\tilde{\mathbf B}=\mathbf B\mathbf U_Q$, which appears in the $\mathrm{PSF}$. Here, $\tilde{\mathbf B}$ is obtained by applying a $Q$-point DFT to each row of $\mathbf B$. Since every row of a feasible $\mathbf B$ must contain exactly $K$ ones and $Q-K$ zeros, we impose a circular-shift structure on $\mathbf B$: each row is chosen as a circularly shifted version of a common $1\times Q$ binary vector $\mathbf b$ with exactly $K$ ones. This restriction reduces the design space, but it facilitates our  construction by reducing the design variables to $\mathbf b$ and the circular shifts across rows.

\par The circular-shift relationship across the rows of $\mathbf B$ induces a special structure in its 2D-DFT, equivalently the $\mathrm{PSF}$. To explain this, we define the $p^{\mathrm{th}}$ row of $\mathbf{B}$ as the $r[p]$ right circular shift of the $1\times Q$ vector $\mathbf{b}$. The 1D-DFT of $\mathbf{b}$ is defined as $\tilde{\mathbf{b}}=\mathbf{b} \mathbf{U}_Q$. As circularly shifting a vector induces linear phase modulation in the DFT \cite{oppenheim1999discrete}, it follows that the $p^{\mathrm{th}}$ row within $\tilde{\mathbf{B}}= \mathbf{B} \mathbf{U}_Q$ in our construction is
\begin{align}
    \!\!\tilde{\mathbf{B}}(p,:) \!&=\! \tilde{\mathbf{b}} \odot [1, e^{-j\frac{2\pi r[p]}{Q}}, e^{-j\frac{4\pi r[p]}{Q}}, \cdots e^{-j\frac{2\pi r[p](Q-1)}{Q}}],
\end{align}
where $\odot$ denotes the Hadamard product. The matrix $\tilde{\mathbf{B}}$ is then
\begin{equation}
\label{eq:Btilde_basic}
\tilde{\mathbf{B}} \!=\!\!
    \begin{bmatrix}
1 & \tikzmarknode{g1}{e^{-j\frac{2\pi r[1]}{Q}}} & e^{-j\frac{4\pi r[1]}{Q}} 
& \cdots & e^{-j\frac{2\pi r[1](Q-1)}{Q}} \\
1 & \tikzmarknode{g2}{e^{-j\frac{2\pi r[2]}{Q}}} & e^{-j\frac{4\pi r[2]}{Q}} 
& \cdots & e^{-j\frac{2\pi r[2](Q-1)}{Q}} \\
\vdots & \tikzmarknode{gvdots}{\vdots} & \vdots & \ddots & \vdots \\
1 & \tikzmarknode{gP}{e^{-j\frac{2\pi r[P]}{Q}}} & e^{-j\frac{4\pi r[P]}{Q}} 
& \cdots & e^{-j\frac{2\pi r[P](Q-1)}{Q}}
\end{bmatrix} \!\!\mathrm{diag}(\tilde{\mathbf{b}})
\begin{tikzpicture}[overlay,remember picture]
\node[
    draw,
    dashed,
    rounded corners,
    inner xsep=3pt,
    inner ysep=3pt,
    fit=(g1)(gP)
] (gbox) {};
\node at ([yshift=-7pt]gbox.south) {$\mathbf g$};
\end{tikzpicture}.
\end{equation}
The vector containing the twiddle factors induced by row-wise circular shifts in our construction is defined as
\begin{equation}
    \mathbf{g}
=
\left(
e^{-j\frac{2\pi r[1]}{Q}},
e^{-j\frac{2\pi r[2]}{Q}},
\ldots,
e^{-j\frac{2\pi r[P]}{Q}}
\right)^T .
\end{equation}
Using this definition, \eqref{eq:Btilde_basic} can be written in compact form as
\begin{equation}
\label{eq:Btilde_simp}
\tilde{\mathbf{B}}
=
[\mathbf{g}^0,\mathbf{g}^1,\ldots,\mathbf{g}^{Q-1}]
\operatorname{diag}(\tilde{\mathbf{b}}),
\end{equation}
where $\mathbf{g}^q$ is the vector obtained by raising each element of
$\mathbf{g}$ to the power $q$. Now, $\mathcal{F}_{\mathrm{2D}}(\mathbf{B})
=
\mathbf{U}_P\tilde{\mathbf{B}}$ can be written as
\begin{equation}
\label{eq:F2D_simplified}
\mathcal{F}_{\mathrm{2D}}(\mathbf{B})
=
[\mathbf{U}_P\mathbf{g}^0,\mathbf{U}_P\mathbf{g}^1,\ldots,
\mathbf{U}_P\mathbf{g}^{Q-1}]
\operatorname{diag}(\tilde{\mathbf{b}}).
\end{equation}
From \eqref{eq:F2D_simplified}, we observe that the $q$th column of
$\mathcal{F}_{\mathrm{2D}}(\mathbf{B})$, equivalently of the scaled $\mathrm{PSF}$,
is the 1D-DFT of $\mathbf{g}^{q-1}$ multiplied by $\tilde{b}[q]$.  

\par To obtain a flat $\mathrm{PSF}$ magnitude outside the first column, the columns of $\mathcal{F}_{\mathrm{2D}}(\mathbf B)$ indexed by $q\in\{2,3,\ldots,Q\}$ must satisfy two conditions. First, these columns must have the same $\ell_2$ norm. Second, the energy within each of these columns must be distributed uniformly across the $P$ rows. Under our circular-shift construction, \eqref{eq:F2D_simplified} shows that the $q$th column of $\mathcal{F}_{\mathrm{2D}}(\mathbf B)$ is $\tilde b[q]\mathbf U_P\mathbf g^{q-1}$, whose $\ell_2$ norm is $ |\tilde b[q]|\left\|\mathbf U_P\mathbf g^{q-1}\right\|_2 = |\tilde b[q]|\sqrt{P}$. This follows from the fact that $\mathbf U_P$ is unitary and that the entries of $\mathbf g$ have unit modulus. Therefore, the first condition is satisfied by choosing $\mathbf b$ such that the DFT coefficients $\{\tilde b[q]\}_{q=2}^{Q}$ have equal magnitude. Once such a $\mathbf b$ is fixed, the second condition can be addressed by designing $\mathbf g$, or equivalently the circular shifts $\{r[p]\}_{p=1}^{P}$, so that $\mathbf U_P\mathbf g^{q-1}$ has a flat magnitude profile $\forall q \in\{2,3,\ldots,Q\}$.

\par We discuss how to construct $\mathbf b$ so that its non-DC DFT magnitude is flat. Let $\mathcal S=\{s_1,\ldots,s_K\}$ be the support of $\mathbf b$, i.e., $b[i]=1$ if $i\in\mathcal S$ and $0$ otherwise. Then, its DFT entry
\begin{equation} 
\tilde b[q] =\frac{1}{\sqrt{Q}}\sum_{k=1}^{K} e^{-j\frac{2\pi(q-1)s_k}{Q}} . 
\end{equation} 
This is the same sum in \cite[Eq.~(10)]{Welch_diff_sets}, up to normalization and conjugation. If $\mathcal S$ is a cyclic difference set, then \cite{Welch_diff_sets}
\begin{equation} |\tilde b[q]| = \sqrt{\frac{K(Q-K)}{Q(Q-1)}}, \qquad \forall q \in \{2,\ldots,Q \}. 
\end{equation}
A cyclic $(Q,K,\lambda)$ difference set is a $K$-element subset in  $\mathbb Z_Q=\{1,2,\cdots Q\}$ such that the cyclic differences
$(s_a-s_b)\bmod Q $ $\forall s_a,s_b\in\mathcal S$, contain every nonzero element of $\mathbb Z_Q$ exactly $\lambda$ times \cite{baumert2006cyclic}. Thus, setting $\mathbf b$ to an indicator on a cyclic difference set ensures equal-magnitude non-DC DFT coefficients. Since such sets can exist only if $K(K-1)=\lambda(Q-1)$, our construction requires $Q-1$ to divide $K(K-1)$.

\par The design of circular shifts $\{r[p]\}_{p=1}^{P}$ so that $\mathbf U_P\mathbf g^{q-1}$ has a flat magnitude profile $\forall q\in\{2,3,\ldots,Q\}$ was studied in \cite[Sec. IV]{kumari2021adaptive}. The key idea in \cite{kumari2021adaptive} is to set $\mathbf g$ to a Zadoff-Chu (ZC) sequence \cite{chu1972polyphase}, whose DFT has a constant magnitude. Moreover, for a prime-length ZC sequence $\mathbf{g}$ with root $1$, its exponent $\mathbf g^{q-1}$ also has constant DFT magnitude $ \forall q \in \{2,\ldots,Q \}$. Setting $\mathbf{g}$ to such a ZC sequence gives
\begin{equation}
r[p]=\left[\frac{p(p-1)}{2}\right]_{\mathrm{mod}\,P},
\qquad p=1,2,\ldots,P .
\label{eq:ZC-shifts}
\end{equation}
The design in \cite{kumari2021adaptive} assumes $P=Q$ with $P$ prime (as in our design), due to the use of prime-length ZC codes for $\mathbf{g}$. Finally, our method relaxes the $K=1$ restriction with \cite{kumari2021adaptive} by using a cyclic difference set-based $\mathbf{b}$. An example of our construction for $(P,Q,K)=(7,7,3)$ is shown in Fig. \ref{fig:sampling_structures}(a).

\section{Numerical Results}
\begin{figure*}[t]
    \centering
    \begin{minipage}{0.3\textwidth}
        \centering
        \includegraphics[width=\linewidth,trim={0cm 0.25cm 0cm 0.8cm},clip]{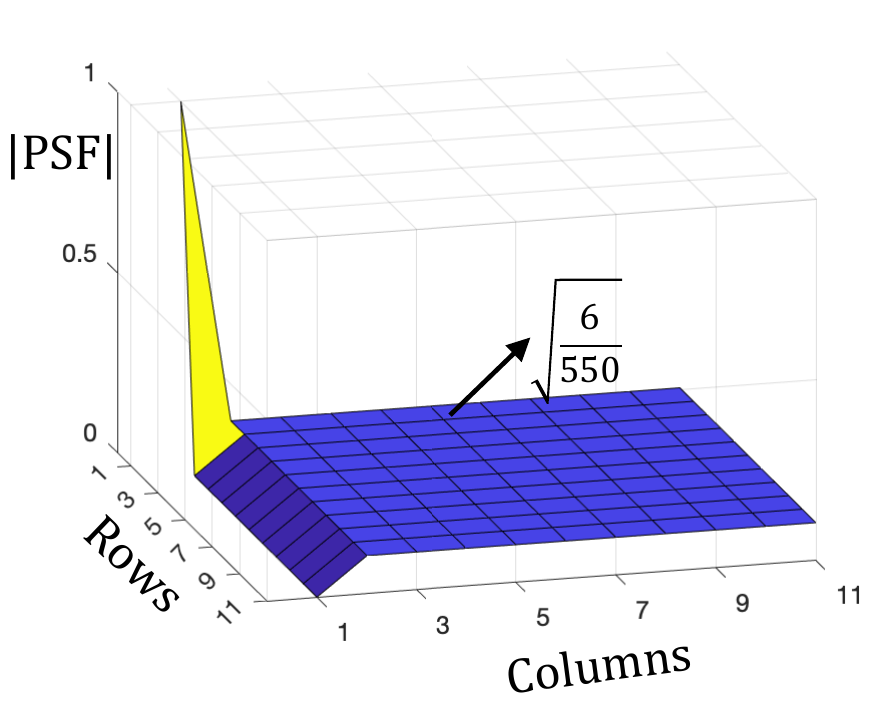}
        \caption*{(a) PSF magnitude for $11\times 11$, $\,K=5$}
    \end{minipage}
    \hfill
    \begin{minipage}{0.3\textwidth}
        \centering
        \includegraphics[width=\linewidth,trim={1cm 6.5cm 1cm 7cm},clip]{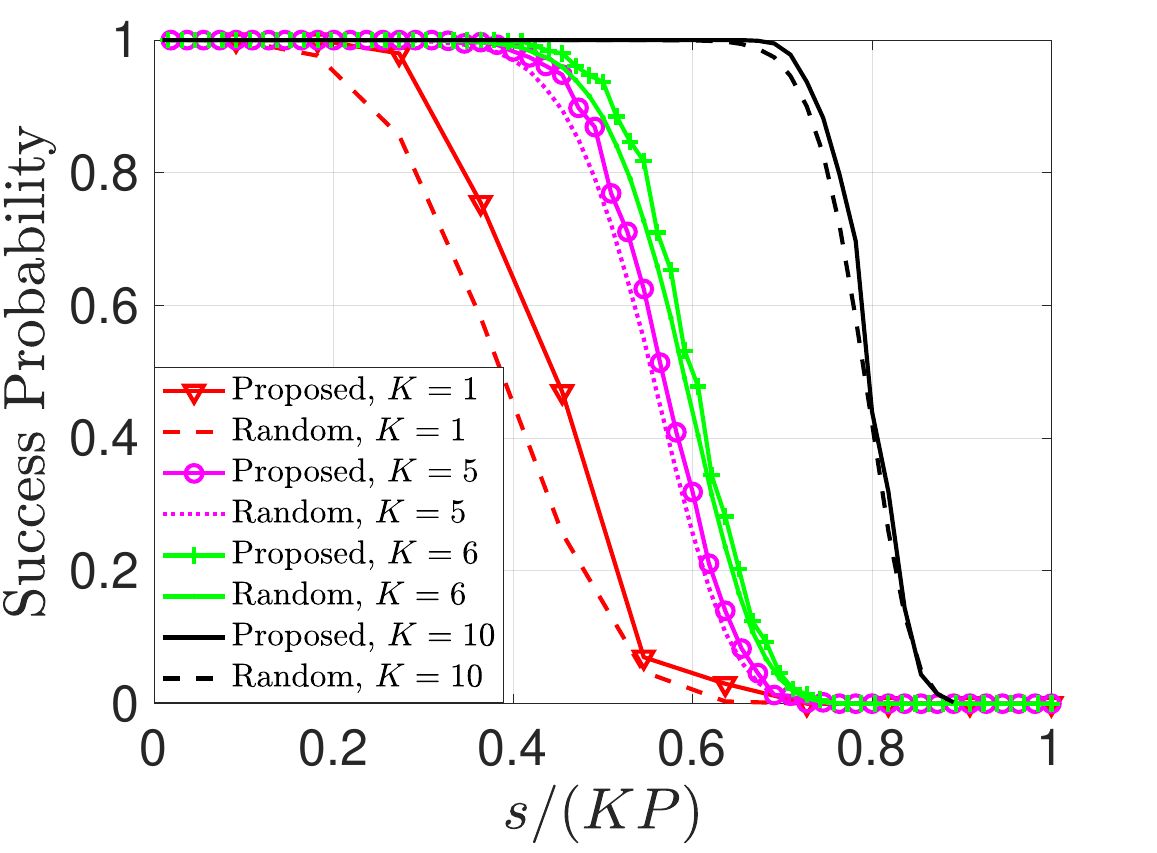}
        \caption*{(b) Signal recovery rate for $11\times 11$.}
    \end{minipage}
    \hfill
    \begin{minipage}{0.3\textwidth}
        \centering
        \includegraphics[width=\linewidth,trim={1cm 6.5cm 1cm 7cm},clip]{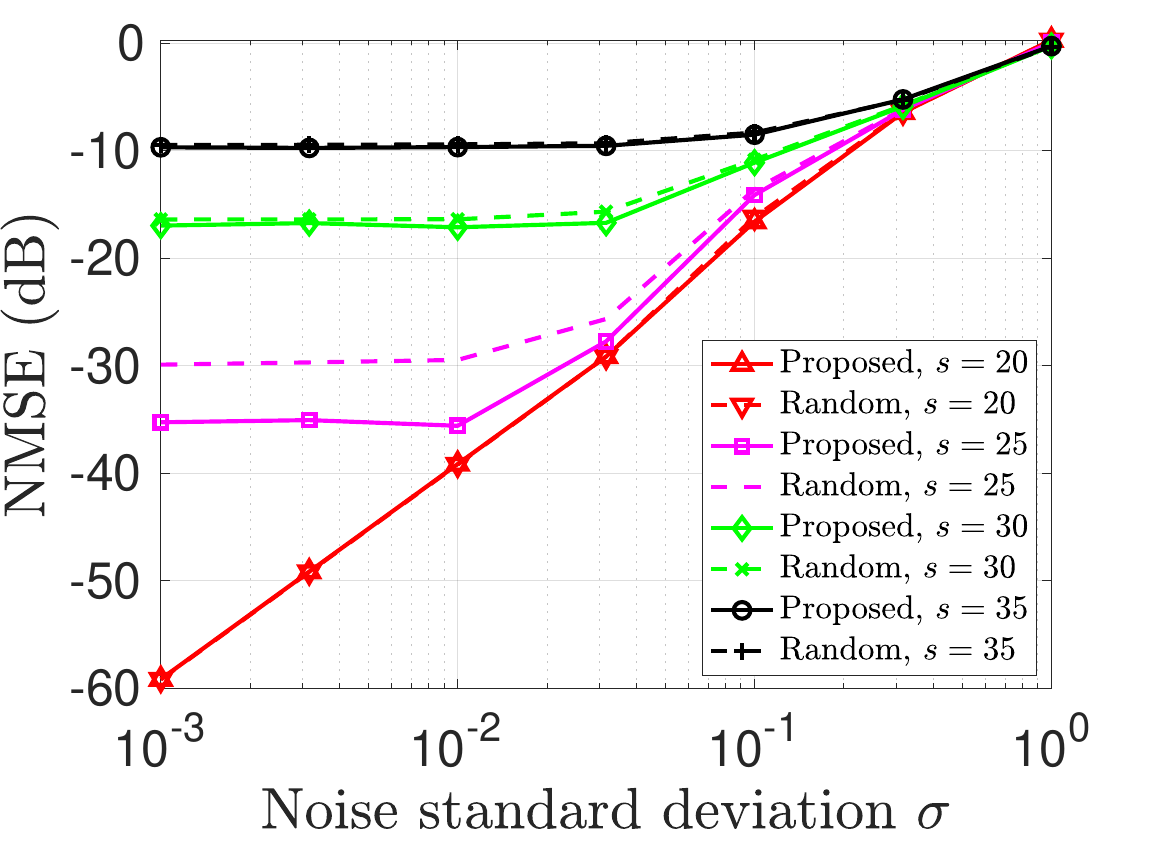}
        \caption*{(c) NMSE vs $\sigma$ for $11 \times 11$, $K=5$.}
    \end{minipage}
    \caption{(a) Our construction for $\mathbf{B}$ yields a PSF that attains the bound in \eqref{eq:bound_prop}. Here, $P=Q=11$ and $K=5$. (b) In a noiseless setting, i.e., $\sigma=0$, subsampling the Fourier representation according to our construction results in better sparse recovery than a comparable feasible random scheme, especially for $K=1$. (c) For low noise standard deviation $\sigma$, subsampling-induced artifacts dominate the sparse recovery performance, and our construction achieves a clear improvement for intermediate sparsity levels, i.e., around $s=25$. For a small number of non-zeros coefficients $s$, recovery is easy for both schemes, while for a large $s$, recovery is difficult for both. Finally, for a large $\sigma$, the gap between the proposed and random schemes diminishes.}
    \label{fig:Simulation_results}
\end{figure*}
\par We consider a 2D-CS problem of recovering a sparse matrix of size $11 \times 11$, i.e., $P=Q=11$, from compressed Fourier measurements acquired under a per-row readout budget constraint. We generate $1000$ random realizations of an $s$-sparse matrix $\mathbf{X}$, where the support indices are drawn uniformly at random from the $11 \times 11$ grid. The $s$-nonzero coefficients are drawn according to the random fading model $(0.5+\delta^2)e^{j\phi}$ \cite{tang2013compressed}, where $\delta\sim\mathcal{N}(0,1)$ controls the amplitude and $\phi\sim\mathcal{U}[0,2\pi)$ is an independent random phase. In this setup, our construction yields subsampling patterns only for $K \in \{1,5,6,10\}$, since cyclic $(Q,K,\lambda)$ difference sets exist only for these values of $K$ when $Q=11$ \cite{baumert2006cyclic}. For illustration, we show $|\mathrm{PSF}|$ corresponding to our construction for $K=5$ in Fig. \ref{fig:Simulation_results}(a). We observe that the $\mathrm{PSF}$  has a flat magnitude outside the first column and it attains our bound  in \eqref{eq:bound_prop}.   

\par In our simulations, sparse recovery is performed using the OMP for a known sparsity $s$. We compare two subsampling schemes: (i) the proposed deterministic scheme, which selects entries from cyclically shifted difference-set locations across rows, with the cyclic shifts defined in \eqref{eq:ZC-shifts}, and (ii) a random scheme, where $K$ entries are chosen uniformly at random from each row and independently across rows. For each realization of $\mathbf{X}$, reconstruction with the random scheme is performed for $50$ independent random subsampling configurations satisfying the same per-row budget constraint. To study the advantage of the proposed design over the random scheme, we first consider a noiseless setting. We report 
the probability of successful signal recovery, where recovery is declared successful if $\Vert \mathbf{X}-\hat{\mathbf{X}}\Vert_{\mathrm{F}}/\Vert \mathbf{X}\Vert_{\mathrm{F}}<10^{-3}$. This result, shown in Fig. \ref{fig:Simulation_results}(b), is sketched against the sparsity ratio $s/(PQ)$ for different readout budgets. We observe that our construction outperforms random subsampling, with a particularly significant improvement when the readout budget is small, i.e., $K=1$. As $K$ increases, the gap diminishes because random subsampling is less likely to produce structured patterns that result in a large mutual coherence $\mu$. For example, unfavorable subsampling configurations, such as sparse repetitive readout across rows which induces high mutual coherence, become unlikely for a large $K$. Although random subsampling under a per-channel budget achieves comparable performance for larger readout budgets, i.e., $K\in\{5,6,10\}$, the proposed construction for $\mathbf{B}$ provides a deterministic alternative that is guaranteed to attain our mutual coherence bound in \eqref{eq:bound_prop}. 

\par Finally, we consider a noisy setting and compare the normalized mean squared error (NMSE) in the sparse estimate using the designed and random subsampling configurations. Here, the compressed measurements in \eqref{eq:Fourier_meas_vec} are perturbed by Gaussian noise, with $\mathbf{v}$ drawn from $\mathcal{CN}(0,\sigma^2)$. The stopping rule in the OMP terminates the algorithm when the residue's norm falls below $\sqrt{M}\sigma$ or when $M$ iterations are reached, whichever occurs first. Fig. \ref{fig:Simulation_results}(c) shows that, at low $\sigma$ and an intermediate $s\approx 25$, the NMSE with the proposed design is lower than that associated with feasible random subsampling under the same per-row readout budget. In this regime, subsampling-induced aliasing ($\mathrm{PSF}\circledast \mathbf{X}$) dominates Gaussian noise ($\mathbf{A}^{\ast}\mathbf{v}$), thereby making the benefit of minimizing $\mu$ more pronounced. At a large $\sigma$, the performance gap diminishes because the noise becomes comparable to aliasing in OMP's matching step.
\section{Conclusions}
We studied 2D-Fourier compressed sensing of a sparse signal under a strict per-channel readout budget, where only a few entries can be acquired from each row of its Fourier representation. We showed that this constraint leads to a mutual coherence lower bound that is strictly larger than the classical Welch bound, for the same total number of measurements. We then proposed a deterministic readout construction based on cyclic difference sets and Zadoff-Chu-based circular shifts, which attains the derived bound for a class of problem dimensions. Numerical results demonstrate that the proposed constructions achieve better sparse recovery than comparable feasible random subsampling schemes, particularly when the per-channel readout budget is small. 
\bibliographystyle{IEEEtran}        
\bibliography{ref.bib}

\end{document}